\documentclass[11pt,a4paper]{article}
\pdfoutput=1
\usepackage[margin=1in]{geometry}

\usepackage[utf8]{inputenc}
\usepackage[english]{babel}
\usepackage{microtype}
\usepackage{xcolor}
\usepackage{graphicx}
\usepackage{amsmath}
\usepackage{amsthm}
\usepackage{amssymb}
\usepackage[unicode]{hyperref}
\usepackage{doi}
\usepackage{framed}
\usepackage{cleveref}
\usepackage{xspace}
\usepackage{enumerate}
\usepackage{thm-restate}

\usepackage{algorithm}
\usepackage{algorithmicx,algpseudocode}
\algblockdefx{StartRepeat}{EndRepeat}{\textbf{repeat}}{}
\algnotext{EndRepeat}

\definecolor{darkgreen}{rgb}{0,0.5,0}
\hypersetup{
    colorlinks=true,
    linkcolor=red,
    citecolor=darkgreen,
    filecolor=black,
    urlcolor=[rgb]{0,0.1,0.5},
    pdftitle={Tight Bounds for Deterministic High-Dimensional Grid Exploration},
    pdfauthor={Sebastian Brandt, Julian Portmann and Jara Uitto}
}

\newtheorem{theorem}{Theorem}[section]

\newtheorem{lemma}[theorem]{Lemma}

\newtheorem{observation}[theorem]{Observation}
\newtheorem{assumption}[theorem]{Assumption}
\newtheorem*{theorem*}{Theorem}

\newcommand{\mult}{\textsc{Mult}}
\newcommand{\multback}{\textsc{Multback}}

\newcommand{\fin}{\text{fin}}
\newcommand{\north}{\text{north}}
\newcommand{\south}{\text{south}}
\newcommand{\stay}{\text{stay}}
\newcommand{\exsync}{\textsc{Explore}}
\newcommand{\incstack}[1]{\textsc{IncreaseStackSize}(\ensuremath{#1})}
\newcommand{\initstack}[1]{\textsc{InitializeStackSize}(\ensuremath{#1})}
\newcommand{\movestack}[1]{\textsc{MoveStack}(\ensuremath{#1})}
\newcommand{\isdiv}[1]{\textsc{IsDivisible}(\ensuremath{#1})}
\newcommand{\multstack}[1]{\textsc{MultiplyStackSize}(\ensuremath{#1})}
\newcommand{\divstack}[1]{\textsc{DivideStackSize}(\ensuremath{#1})}
\newcommand{\followroute}[1]{\textsc{FollowRoute}(\ensuremath{#1})}
\newcommand{\indeg}[1]{\ensuremath{\text{indegree}({#1})}}
\newcommand{\virt}{\ensuremath{\text{Virt}}}
\newcommand{\cell}[1]{\ensuremath{\text{Cell}({#1})}}
\newcommand{\level}[1]{\ensuremath{\text{Level}({#1})}}

\newcommand{\myemail}[1]{\,$\cdot$\, {\small #1}}
\newcommand{\myaff}[1]{\,$\cdot$\, {\small #1}\par\medskip}

\newenvironment{myabstract}
{\list{}{\listparindent 1.5em%
        \itemindent    \listparindent
        \leftmargin    1cm
        \rightmargin   1cm
        \parsep        0pt}%
    \item\relax}
{\endlist}

\newenvironment{mycover}
{\list{}{\listparindent 0pt
        \itemindent    \listparindent
        \leftmargin    1cm
        \rightmargin   1cm
        \parsep        0pt}%
    \raggedright
    \item\relax}
{\endlist}

\begin{document}



\begin{mycover}
  {\huge\bfseries\boldmath Tight Bounds for Deterministic High-Dimensional Grid Exploration \par}
\bigskip
\bigskip
\bigskip

\textbf{Sebastian Brandt}
\myemail{brandts@ethz.ch}
\myaff{ETH Zurich}

\textbf{Julian Portmann}\footnote{Supported by the Swiss National Foundation, under project number 200021\_184735.}
\myemail{pjulian@ethz.ch}
\myaff{ETH Zurich}

\textbf{Jara Uitto}
\myemail{jara.uitto@aalto.fi}
\myaff{Aalto University}

\bigskip

\end{mycover}
\begin{myabstract}
\noindent\textbf{Abstract.}
	We study the problem of exploring an oriented grid with autonomous agents governed by finite automata.
	In the case of a $2$-dimensional grid, the question how many agents are required to explore the grid, or equivalently, find a hidden treasure in the grid, is fully understood in both the synchronous and the semi-synchronous setting.
	For higher dimensions, Dobrev, Narayanan, Opatrny, and Pankratov [ICALP'19] showed very recently that, surprisingly, a (small) constant number of agents suffices to find the treasure, independent of the number of dimensions, thereby disproving a conjecture by Cohen, Emek, Louidor, and Uitto [SODA'17].
	Dobrev et al.~left as an open question whether their bounds on the number of agents can be improved.
	We answer this question in the affirmative for deterministic finite automata: we show that $3$ synchronous and $4$ semi-synchronous agents suffice to explore an $n$-dimensional grid for any constant $n$.
	The bounds are optimal and notably, the matching lower bounds already hold in the $2$-dimensional case.

	Our techniques can also be used to make progress on other open questions asked by Dobrev et al.: we prove that $4$ synchronous and $5$ semi-synchronous agents suffice for \emph{polynomial-time} exploration, and we show that, under a natural assumption, $3$ synchronous and $4$ semi-synchronous agents suffice to explore \emph{unoriented} grids of arbitrary dimension (which, again, is tight).
\end{myabstract}

\thispagestyle{empty}
\setcounter{page}{0}
\newpage

\section{Introduction}
\label{sec:intro}

Grid search by mobile agents is one of the fundamental primitives in swarm robotics and a natural abstraction of foraging behavior of animals.
For example in the case of cost-efficient robots or insects, a single agent has relatively limited computation and communication capabilities and hence, many independent agents are required to efficiently solve tasks.
To understand such collective problem solving better, knowledge from distributed computing has proven valuable.
For instance, Feinerman et al.\ gave tight bounds on the time complexity of a collective grid search problem inspired by desert ants~\cite{Feinerman2012}.
In this paper, we focus on the minimum \emph{number} of agents required to solve the grid search problem.
A series of papers \cite{Emek15, Cohen2017, Brandt2018} nailed down the exact complexity of the $2$-dimensional case, that is, discovered the exact number of synchronous/semi-synchronous and deterministic/randomized finite automata needed to explore a $2$-dimensional grid.
However, the approaches in these works do not generalize (well) to higher dimensions.
The only known tight bound achieved by such a generalization is obtained by the recent protocol for the deterministic semi-synchronous 3-dimensional setting by Dobrev, Narayanan, Opatrny, and Pankratov~\cite{Dobrev2019}.

The authors of \cite{Dobrev2019} also gave a more general result:
they showed how to implement a stack data structure using only a constant number of agents governed by finite automata.
By employing this stack in their search protocols, they show how to explore an $n$-dimensional grid using only a (small) constant number of agents, for any positive integer $n$.
In particular, the number of agents is independent of the dimension $n$.

For the case of a $2$-dimensional grid the required number of agents is fully understood.
However, for higher dimensions there are still gaps between the best upper and lower bounds.
Indeed, Dobrev et al.\ left as open questions the tight complexities of exploring high-dimensional grids in the synchronous/semi-synchronous and deterministic/randomized settings.
In this work, we answer these questions for the deterministic setting.
Moreover, building on our techniques we make progress on other open questions by Dobrev et al.

\subsection{Results and Techniques}
Similarly to the approach by Dobrev et al.~\cite{Dobrev2019}, our search protocols rely on an efficient implementation of a stack data structure.
One agent is dedicated to do the actual search
while the remaining agents implement a stack (together with the searching agent that indicates the base of the stack) with their positions on the grid.
On a high level, the size of the stack encodes the cell the searching agent is supposed to explore next, relative to the current position of the searching agent.
Both our protocol and the protocol from \cite{Dobrev2019} explore the grid by repeatedly reading the stack, moving the searching agent to the cell indicated by the stack, moving the searching agent back to its original cell, and incrementing the stack.
The difficult part is to be able to effectively read the stack (without destroying the stack in the process) despite the fact that the size of the stack grows arbitrarily far beyond the number of states in the finite automaton reading the stack.
The authors of \cite{Dobrev2019} managed to implement this data structure using $4$ agents in the synchronous and $5$ agents in the semi-synchronous setting and showed how to explore oriented grids with as many agents.

One of our main contributions is to implement this stack and the operations required for reading it with only $3$ (synchronous) agents (including the searching agent), which is optimal given the grid exploration lower bound by Emek et al.~\cite{Emek15}.
We achieve this by a careful design of an encoding scheme that transforms the location of a cell to be explored into a single integer (that can be represented by the stack size) by interpreting the coordinates of the cell (relative to the current location of the searching agent) as exponents of distinct prime factors.
One crucial advantage of this specific encoding is that there is a way to read the stack (using $3$ synchronous agents), i.e., to repeatedly provide the searching agent with different parts of the encoded information, that does not destroy the encoded information, but instead changes the encoding slightly: replacing the base (prime) for one of those exponents by a different prime (and then switching to the next base prime).
The technical details why such replacement operations can be performed by $3$ synchronous agents and why they allow the searching agent to obtain the desired information are covered in Sections~\ref{sec:blocks} and~\ref{sec:protocol}.
Moreover, by adding one agent as a synchronizer, the protocol can be made to work in the semi-synchronous setting.

\begin{restatable}{theorem}{oriented}\label{thm:oriented}
	For any positive integer $n$, the $n$-dimensional (oriented) grid can be explored by $3$ synchronous finite automata, resp.\ $4$ semi-synchronous finite automata.
\end{restatable}

\paragraph{Unoriented Grids.}
An underlying assumption of the setting considered so far is that the agents are aware of the $2n$ cardinal directions, i.e., they know for each of the $n$ dimensions of the grid which two adjacent cells are neighbors in that dimension, and each dimension is oriented.
Or, to put it simply, the agents know which directions are north, south, etc.; in particular the directions are globally consistent.
In contrast, in the \emph{unoriented} setting considered in~\cite{Dobrev2019}, each cell is endowed with a labeling that indicates for each cell which neighbor is north, south, etc.\ (and for each of the $2n$ directions there is exactly one neighbor), but there is no consistency guarantee between the directions indicated by the labels of different cells.
In fact, the setting does not preclude that after going north twice, you do not end up in the cell in which you started.

In their work, Dobrev et al.\ also ask ``How many additional agents are necessary to solve the problem in \textit{unoriented grids}?''.
We show, perhaps surprisingly, that the unoriented case is no harder than the oriented case given the following (natural) assumption:
If we follow some fixed direction, we never end up back in the same cell where we started.

\newcounter{unoricounter}
\setcounter{unoricounter}{\value{theorem}}
\newcounter{unorisection}
\setcounter{unorisection}{\value{section}}
\begin{theorem}[Simplified]\label{thm:unori}
	Under a natural assumption, for any positive integer $n$, $3$ synchronous finite automata, resp.\ $4$ semi-synchronous finite automata, suffice to explore any $n$-dimensional unoriented grid.
\end{theorem}

The key idea to obtain Theorem~\ref{thm:unori} is that, even without a globally consistent orientation, we can implement a (virtual) stack.
Due to the missing consistency, the same cell may occur repeatedly in the stack, but we can show that we can bound the number of occurrences for each cell and that the agents can distinguish between the different occurrences of the same cell.
In essence, we will show that the stack corresponds to (a part of) a DFS exploration of an infinite tree consisting of those edges (between cells) that point north.

\paragraph{Polynomial Time Protocol.}
The task of exploring the entire grid can equivalently be described as finding a treasure located at some distance $D$ from the starting point.
This formulation allows us to discuss the \emph{efficiency} of a protocol, i.e., its runtime with respect to $D$.
We observe that our encoding scheme for the oriented grid using only $3$ synchronous, resp.\ $4$ semi-synchronous, agents might result in exponential time.
However, we show that with one additional agent, certain stack operations can be extended to work for non-constant values.
This allows us to use a different exploration scheme, proposed by Dobrev et al.~\cite{Dobrev2019}, which is similar to the well-known spiral search, resulting in a polynomial runtime.
\begin{restatable}{theorem}{thmpoly}\label{thm: poly}
	For any positive integer $n$, the $n$-dimensional (oriented) grid can be explored by:
	(1) 4 synchronous agents in time $O(V(D)^2)$, and
	(2) 5 semi-synchronous agents in time $O(V(D)^3)$,
	where $V(D) = \Theta(D^n)$ is the volume of the $\ell_1$-ball of radius $D$.
\end{restatable}

\newcommand{\sync}{\ensuremath{\mathcal{FSYNC}}}
\newcommand{\ssync}{\ensuremath{\mathcal{SSYNC}}}
\newcommand{\async}{\ensuremath{\mathcal{ASYNC}}}

\subsection{Further Related Work}
In a typical graph exploration setting, we are given a graph where initially, one or more mobile agents are placed on some vertices of the graph.
The agents are able to traverse along the edges and their goal is to \emph{explore} the graph, that is, visit every node or edge of the graph.
Equivalently, one can think of searching for a treasure hidden on an edge or a node of the graph.
Graph exploration has been widely studied in the literature (see, for example, \cite{Rollik1979, Panaite1998, Deng1999, Albers2000, Diks2004, FraigniaudIPPP2005}) and it comes in many variants.

A classic setting is the \emph{cow-path} problem, where a single agent, the cow, is searching for adversarially hidden food on a path~\cite{Beck1964, Baeza-Yates1993}.
The goal for the cow is to minimize the number of edge-traversals until the food is found.
It is known that a simple spiral search is optimal and this algorithm also generalizes to the case of grids.
This problem was also studied in the case of many cows~\cite{Lopez-Ortiz2001}.
Closely related to our work is the exploration of \emph{labyrinths}, i.e., $2$-dimensional grids where some cells are blocked~\cite{Budach1978}.
It is known that two finite automata or one automaton with two pebbles (movable marker) suffice for co-finite labyrinths, where a finite amount of cells are not blocked~\cite{Blum1978}.
Finite labyrinths, where a finite amount cells are blocked, can be explored with one automaton and four pebbles, whereas one automaton and one pebble is not enough~\cite{Blum1977, Hoffmann1981}.
An agent with $\Theta(\log \log n)$ pebbles can explore all graphs and this bound is tight~\cite{Disser2016}.

In the case of many agents, the agents typically operate in \emph{look-compute-move} cycles.
First, the agents take a local snapshot, then decide on the next operation, and finally, execute the operation.
Graph exploration can be divided into \emph{synchronous} (\sync{}), \emph{semi-synchronous} (\ssync{}), and \emph{asynchronous} (\async{}) variants~\cite{Sugihara1996, Suzuki96, Suzuki1999}.
In the \sync{} setting, the execution is divided into synchronous rounds, where in every round, every agent executes one cycle.
The execution in \ssync{} consists of discrete time steps, where in each step, a subset of the agents executes one atomic cycle.
In the \async{} setting, the cycles are not (necessarily) atomic.
In this paper, we study the \sync{} and the \ssync{} settings.

For finite graphs, a \emph{random walk} provides a simple algorithm that explores the graph in polynomial time~\cite{Aleliunas1979}.
In the case of an infinite $n$-dimensional grid, a random walk finds the treasure with probability $1$ if $n \leq 2$.
However, the expected hitting time, i.e., the time to find the treasure, tends to infinity.
Cohen et al.\ showed that even for the case of two (collaborating) randomized agents governed by finite automata, one cannot achieve any finite hitting time for $n \geq 2$~\cite{Cohen2017}.
Very recently, Dobrev et al.\ showed that $3$ randomized \sync{} and $4$ randomized \ssync{} agents suffice to achieve a finite hitting time for any $n$~\cite{Dobrev2019}.
In this work, we achieve the same bounds with deterministic agents.

This work follows a series of papers inspired the work by Feinerman et al., where they studied the time it takes to find a treasure in a $2$-dimensional grid by $k$ non-communicating agents governed by Turing machines~\cite{Feinerman2012}.
They showed that the time complexity of this task is $\Theta(D^2/k + D)$, where $D$ is the distance from the origin to the treasure.
This bound can be matched by finite automata that are allowed to communicate within the same cell~\cite{Emek2014}.
Emek et al.\ asked what is the minimum number of agents required to find the treasure~\cite{Emek15}.
They showed that at least $3$ synchronous deterministic agents are required and that $3$ synchronous deterministic, $4$ semi-synchronous deterministic, and $3$ semi-synchronous randomized agents are enough.
Cohen et al.\ \cite{Cohen2017}\ and Brandt et al.\ \cite{Brandt2018}\ showed the matching lower bounds for the randomized and deterministic semi-synchronous cases, respectively.

\section{Preliminaries}

\paragraph{Grids.}
We consider the problem of exploring the infinite $n$-dimensional grid, whose vertices are the elements of $\mathbb{Z}^n$, which we refer to as \emph{cells}.
A cell $c = (c_1, \dots, c_i, \dots, c_n)$ is described by its coordinates and two cells $c$ and $c'$ are adjacent (i.e., connected by an edge) if they differ in one coordinate by $1$, i.e., there is a dimension $i$ such that $|c_i - c'_i| = 1$ and $c_j = c'_j$ for $j \neq i$.
When talking about distance, we will use the $\ell_1$ or \emph{Manhattan distance}, which is defined as $d(c, c') = \sum_i |c_i - c'_i|$.

In the \emph{oriented} case, we assume that there is a consistent labeling of the edges by both of its endpoints, which in the 2-dimensional case can be thought of as the directions of a compass: north, south, east, and west.
In general, an edge $(c, c')$ is labeled by $(+1, i)$ from the side of $c$ (and thus $(-1, i)$ from the side of $c'$) if we have that $c_i + 1 = c'_i$.

For \emph{unoriented} grids, we assume that each endpoint of an edge has a label from $\{1, \dots, 2n\}$.
We will also refer to these labels as the \emph{ports} of a cell.
The only assumption we make is that the labels around each cell are pairwise distinct, i.e., each cell has every port from $1$ to $2n$ exactly once.
Thus, each edge can receive any pair of labels from $\{1, \dots, 2n\}$.

\paragraph{Exploration.}
The exploration is performed by $m$ agents, $a_1, \dots, a_m$, which are initially all placed in the same cell, called the \emph{origin}.
W.l.o.g.\ we assume the origin to have coordinates $(0, \dots, 0)$.
The agents cannot distinguish different cells (including the origin); in particular, they do not know the coordinates of the cell they are in.
Their behavior and movement is controlled by a deterministic finite automaton.
While we require all agents to use the same automaton, they may start in different initial states. (As we only consider protocols with constantly many agents, one can equivalently assume each agent to be controlled by an individual automaton, as we can combine $m$ automata into one by using disjoint state spaces.)
Agents can only communicate if they are in the same cell:
each agent senses the states for which there is an agent that occupies the same cell, and performs its next move and state transition based on this information.
For oriented grids, such a move is described by a direction and dimension.

In the case of unoriented grids, we assume that agents can also see both labels of each incident edge, and perform their decisions based on this information as well.
A move is then described by choosing a port of the current cell and moving along this edge.
Previous work by Dobrev et al. \cite{Dobrev2019} used an essentially equivalent definition:
Each agent could only see the label on its side of each incident edge, but once it arrived in the new cell by traversing some edge, it would obtain the information about the second label on the edge it traversed.
We choose to formalize the model in a slightly different way, as it will simplify the description of our algorithms.
However, we emphasize that for our purposes, the two models can be used interchangeably since within $2n$ steps in the model of Dobrev et al., the agents can learn all information that we assume the agents can immediately see.

Formally, we have a state space $Q$, a transition function $\delta$, and an initial state $q^0_i$ for every agent $a_i$.
For oriented grids, the transition function has the form: $\delta: Q \times 2^Q \to Q \times (\{-1,+1\} \times \{0, 1, \dots, n\})$.
The function maps an agent in state $q \in Q$, which observes the set of states for which there is an agent occupying the same cell, to a new state $q' \in Q$ and a movement, which is described by the direction ($-1$ or $+1$) and the dimension (from $1$ to $n$) along which the agent moves to the respective neighboring cell, where an agent can also choose to \emph{stay} in the same cell which is described by dimension $0$.
We will say that an agent moves \emph{north} if its movement is $(+1, 1)$, and \emph{south} if it is $(-1, 1)$.

For unoriented grids, we change the definition of the transition function slightly to $\delta: Q \times 2^Q \times \{1, \dots, 2n\}^{2n} \to Q \times (\{-1,+1\} \times \{0, 1, \dots, n\})$.
The function maps an agent in state $q \in Q$, which observes both the set of states for which there is an agent occupying the same cell, and, for each port, the other label on the edge corresponding to that port, to a new state $q' \in Q$ and a movement, which is specified by the port via which the agent leaves the current cell, or 0, in which case the agent does not move.

\paragraph{The Schedule.}
Time is divided into discrete units, where in each time step, a set of \emph{active} agents performs a \emph{look-compute-move} cycle.
First, an agent senses the states of all agents in the same cell (and in the case of unoriented grids both of the labels on all incident edges), then it applies the transition function to its own state and all sensed information, and finally it changes its state and moves as indicated by the result.
We assume that one such cycle is atomic, i.e., cycles that start at different times do not overlap.

For the \emph{synchronous} or \sync{} model, we assume that all agents are active at every time step.
We call the system \emph{semi-synchronous}, or the \ssync{} variant, if at every time step only a subset of agents, chosen by an adversary, is active.
While the adversary knows all information about the agents and their behavior, it must schedule each agent infinitely often, to avoid trivial impossibilities.

\paragraph{Exploration Cost.}
Finally, if we discuss the efficiency of a protocol, we consider the following problem, which is equivalent to exploring the grid:
the agents are tasked to find a treasure, which is hidden at some distance $D$ from the origin (without the agents knowing the value of $D$).
This enables us to measure the time or \emph{exploration cost} it takes to find the treasure with respect to $D$.
In the synchronous setting, we measure the exploration cost as the number of time steps needed for an agent to arrive at the cell containing the treasure.
As, in the semi-synchronous model, this number of steps depends on the schedule, we instead define the exploration cost as the total distance traveled by all agents in this setting.

\section{Building Blocks}\label{sec:blocks}

\paragraph{Encoding Information as a Stack.}
Dobrev et al. \cite{Dobrev2019} introduced the idea of using multiple agents to implement a \emph{stack}.
In its simplest form, a stack is just a pair of agents, whose distance encodes some information.
However, to allow for manipulations of the stack, more agents are needed.
Our protocol for exploring $n$-dimensional grids with $3$ synchronous, resp.\ $4$ semi-synchronous, agents will consist of subroutines that involve manipulations of the stack.
The relevant parameter will be the \emph{stack size}, denoted by $X$, which is defined as the distance between the \emph{base} of the stack and the \emph{end} of the stack.
The base of the stack is the location of agent $a_1$, and the end of the stack is the location of the other agents.
We will only be interested in the stack and its size at the very beginning and very end of each subroutine; at these points in time all agents except $a_1$ are guaranteed to be in the same cell, and this cell is guaranteed to be reachable from the cell containing $a_1$ by going repeatedly north, making the notion of a stack well-defined.
Whenever we refer to the base, end, or size of the stack \emph{during} some subroutine, we mean the respective notion at the beginning of the subroutine.

In this section, we will describe the subroutines that form the building blocks of our exploration algorithm.
Moreover, we will show for both the synchronous and the semi-synchronous setting how to implement the subroutines with the desired number of agents.

In \cite{Dobrev2019}, the authors show how to multiply the current stack size by $2$, resp.\ divide it by $2$, using $3$ synchronous agents.
This also provides a way to check whether the current stack size is divisible by $2$.
The idea behind the implementation is simple: while agent $a_1$ stays at the base of the stack, the other two agents, initially located at the end of the current stack, move with different speeds\footnote{An agent moves with speed $1/j$ in some direction if it repeatedly performs the following behavior: first it takes one step in the chosen direction, and then it waits for $j-1$ steps. Note that our speed of $1/j$ is the same as speed $j$ in \cite{Dobrev2019}.}, $a_3$ either away from or towards the base of the stack, and $a_2$ first towards the base, and then reversing direction when the base is reached.
The operation is completed when $a_2$ and $a_3$ meet again (after $a_2$ visited the base).
By choosing a speed of $1$ for $a_2$, and a speed of $1/3$ for $a_3$, and letting move $a_3$ \emph{towards} the base, we achieve that the stack size is halved; by choosing the same speeds and letting move $a_3$ \emph{away from} the base, we achieve that the stack size is doubled.

We will need similar subroutines as building blocks for both our synchronous and semi-synchronous protocols.
More precisely, given a positive integer $k \geq 2$, we want the agents to be able to perform the following operations.
\begin{itemize}
	\item \multstack{k}: Multiply the stack size by $k$.
	\item \isdiv{k}: Check whether the current stack size is divisible by $k$.
	\item \divstack{k}: If the stack size is divisible by $k$, divide the stack size by $k$.
\end{itemize}
We will only require the agents to be able to perform these operations for constantly many $k$, where the constant depends (only) on the dimension $n$ of the grid.

To implement these operations, we simply adapt the protocols for the case $k=2$ from \cite{Dobrev2019} by choosing the speeds of $1/(k-1)$ (instead of $1$) for $a_2$ and $1/(k+1)$ (instead of $1/3$) for $a_3$.
More precisely, we implement the desired operations using $3$ synchronous agents as follows.

\paragraph{\multstack{k}}
While it is usually easier to understand the behavior of an agent if it is described without specifying the exact states and the transition function, we will provide the latter for subroutine \multstack{k} to give an example how to translate the agents' behaviors described in this work into the formal specification of a finite automaton.
Let $k \geq 2$ be a positive integer.
As usual we assume that $a_2$ and $a_3$ are in the same cell $c'$, and $a_1$ is in a cell $c \neq c'$ such that $c'$ can be reached from $c$ by going north repeatedly (i.e., $c$ and $c'$ differ only in the first coordinate, and $c$ has a smaller first coordinate than $c'$).

In subroutine \multstack{k}, we denote the starting state of each agent $a_i$ by $\mult^0_{i, k}$.
Apart from state $\mult^0_{i, k}$, we will use $2k-2$ other states for agent $a_2$, denoted by $\mult^1_{2, k}, \dots, \mult^{k-2}_{2, k}$, $\multback^0_{2, k}, \dots, \multback^{k-2}_{2, k}$ and $\mult^\fin_{2, k}$, and $k+1$ other states for agent $a_3$, denoted by $\mult^1_{3, k}, \dots, \mult^{k}_{3, k}$, and $\mult^\fin_{3, k}$.
Agent $a_1$ always stays in state $\mult^0_{1, k}$ and cell $c$.
Agents $a_2$ moves and changes its state according to the following rules, where ``stay'' indicates that the agents does not move to another cell.
\begin{align*}
	(\mult^0_{2, k}, S)         & \to (\mult^1_{2, k}, \south) \,       & \text{for any } S \in 2^Q \text{ satisfying } \mult^0_{1, k} \notin S  \\
	(\mult^0_{2, k}, S)         & \to (\multback^1_{2, k}, \north) \,    & \text{for any } S \in 2^Q \text{ satisfying } \mult^0_{1, k} \in S     \\
	(\mult^j_{2, k}, S)         & \to (\mult^{j+1}_{2, k}, \stay) \,     & \text{for any } 1 \leq j \leq k-3 \text{ and any } S \in 2^Q           \\
	(\mult^{k-2}_{2, k}, S)     & \to (\mult^0_{2, k}, \stay) \,         & \text{for any } S \in 2^Q                                              \\
	(\multback^0_{2, k}, S)     & \to (\multback^1_{2, k}, \north) \,    & \text{ for any } S \in 2^Q \text{ satisfying } \mult^0_{3, k} \notin S \\
	(\multback^0_{2, k}, S)     & \to (\mult^\fin_{2, k}, \stay) \,      & \text{ for any } S \in 2^Q \text{ satisfying } \mult^0_{3, k} \in S    \\
	(\multback^j_{2, k}, S)     & \to (\multback^{j+1}_{2, k}, \stay) \, & \text{for any } 1 \leq j \leq k-3 \text{ and any } S \in 2^Q           \\
	(\multback^{k-2}_{2, k}, S) & \to (\multback^0_{2, k}, \stay) \,     & \text{for any } S \in 2^Q
\end{align*}
For agent $a_3$, the rules are as follows.
\begin{align*}
	(\mult^0_{3, k}, S) & \to (\mult^1_{3, k}, \north) \quad    & \text{for any } S \in 2^Q \text{ satisfying } \multback^0_{2, k} \notin S \\
	(\mult^0_{3, k}, S) & \to (\mult^\fin_{3, k}, \stay) \quad  & \text{ for any } S \in 2^Q \text{ satisfying } \multback^0_{2, k} \in S   \\
	(\mult^j_{3, k}, S) & \to (\mult^{j+1}_{3, k}, \stay) \quad & \text{for any } 1 \leq j \leq k-1 \text{ and any } S \in 2^Q              \\
	(\mult^k_{3, k}, S) & \to (\mult^0_{3, k}, \stay) \quad     & \text{for any } S \in 2^Q
\end{align*}
The protocol terminates when both $a_2$ and $a_3$ are in states $\mult^\fin_{2, k}$ and $\mult^\fin_{3, k}$, respectively.
The design of the protocol (in particular, of the two rules leading to the two terminal states) ensures that $a_2$ and $a_3$ terminate at the same point in time.
As the rules of the protocol specify that $a_2$ walks with speed exactly $1/(k-1)$, and $a_3$ with speed exactly $1/(k+1)$, we see that the first time $a_2$ and $a_3$ are in the same cell in states $\multback^0_{2, k}$, resp.\ $\mult^0_{3, k}$ (which is the configuration leading to termination in the next step), they are in a cell in distance $kX$ from the base of the stack.
The meeting happens after $a_2$ traversed $(k+1) \cdot X$ cells ($X$ towards the base, $kX$ away from the base), whereas $a_3$ traversed $(k-1) \cdot X$ cells.

\paragraph{\divstack{k}}
Analogously, we can implement division by $k$ by letting $a_3$ walk \emph{towards} the base, instead of away from the base, i.e., by replacing the first rule for $a_3$ by
\[
	(\mult^0_{3, k}, S) \to (\mult^1_{3, k}, \south) \quad \text{for any } S \in 2^Q \text{ satisfying } \multback^0_{2, k} \notin S
\]
while leaving all other rules (for all agents) unchanged.
However, the two rules leading to the terminal states require $a_2$ and $a_3$ to be in states $\multback^0_{2, k}$ and $\mult^0_{3, k}$, respectively, to ensure termination.
If the initial stack size $X$ is divisible by $k$, then the states of the two agents will align perfectly in the cell $c''$ in distance $X/k$ from the base of the stack: after $(k-1)(k+1)$ time steps, $a_2$ has traversed $k+1$ cells with speed $1/(k-1)$, and $a_3$ has traversed $k-1$ cells with speed $1/(k+1)$, hence both are in cell $c''$ in the states leading to the terminal states.
If, however, $X$ is not divisible by $k$, then the states of the two agents do not align when they meet again after $a_2$ visited the base.

\paragraph{\isdiv{k}}
Hence, before dividing by $k$, we will always check whether the current stack size is divisible by $k$.
This can be achieved by having $a_2$ walk towards the base with speed $1$ while increasing a counter modulo $k$ each time it takes a step.
If the counter is at $0$ when $a_2$ reaches $a_1$, the stack size is divisible by $k$; if not, then the stack size is not divisible by $k$.
The subroutine of checking for divisibility by $k$ terminates after $a_2$ has walked back to $a_3$ and informed it whether the current stack size is divisible by $k$ or not.

\paragraph{Further Building Blocks.}
In order to be able to write our synchronous exploration protocol concisely, it will be useful to define a few other subroutines.
As before, we will assume that, in the beginning of the subroutines, agents $a_2$ and $a_3$ will be in the same cell $c'$, representing the end of the stack, and $a_1$ is in a cell $c$ representing the base of the stack that differs from $c'$ only in that its coordinate in dimension $1$ is strictly smaller.
The only exception will be the subroutine \initstack{k} that initializes the stack to some positive integer $k$ by having $a_2$ and $a_3$ walk $k$ steps away from $a_1$---here, all three agents are initially in the same cell.
Apart from \initstack{k}, we define the subroutines \incstack{k} for positive integers $k$, and \movestack{g, i}, where $g \in \{-1, 1\}$ and $i \in \{1, \dots, n\}$.
Subroutine \incstack{k} simply increases the stack size by $k$ (additively) by having $a_2$ and $a_3$ walk $k$ steps away from $a_1$.

A subroutine similar to our \movestack{g, i} was already introduced in \cite{Dobrev2019}.
The purpose of this subroutine is to move the whole stack in some direction specified by dimension $i$ and sign $g$.
In our definition, \movestack{g, i} moves every agent to a new cell that differs from the old cell only by having its $i$th coordinate increased by $g$, i.e., effectively each agent takes one step in dimension $i$.
However, one has to be a bit careful when implementing this subroutine as we want to be able to concatenate it with other subroutines.
In particular, in all other subroutines, agent $a_1$ does not know when the subroutine is started or terminates, while the other agents do know.
In order to also obtain this property for \movestack{g, i}, we implement the desired movement by having $a_2$ walk towards $a_1$, notifying it about the desired step and the chosen direction (upon which $a_1$ performs the step) and then returning to $a_3$, where both $a_2$ and $a_3$ perform the desired step as well.

\paragraph{Semi-Synchronous Agents.}
All of the above subroutines can also be performed by (at most) $4$ semi-synchronous agents, as we show in the following.
Similar to the approach in \cite{Dobrev2019}, we will use one agent ($a_4$) to effectively synchronize the behavior of the other agents, which allows us to essentially execute the $3$-agent synchronous subroutines described above with the remaining $3$ agents.
In more detail, agent $a_4$ will visit the other agents in a suitable order, and each of the other agents will only move when they are in the same cell as $a_4$ (while $a_4$ will not leave the cell of the agent it wants to move next until the agent actually left the cell).
We start by showing how this can be achieved for subroutine \multstack{k}.

As in the synchronous version of the subroutine, we would like the two agents $a_2$ and $a_3$ to move with (relative) speeds $1/(k-1)$ (first towards $a_1$ and, after meeting $a_1$, away from $a_1$) and $1/(k+1)$ (away from $a_1$), respectively, while $a_1$ simply stays at the base of the stack.
The purpose of this design---that when $a_2$ and $a_3$ meet next, they are in a cell that has the $k$-fold distance to $a_1$ as they have currently---can also be achieved by having $a_3$ move $k-1$ steps, then having $a_2$ move $k+1$ steps, and so on, always alternating between the two agents, until they are both in the same cell again (which, by their relative ``speeds'' must have the desired distance to the base of the stack).
This behavior can be ensured by using $a_4$:

Agents $a_2$ and $a_3$ follow their designated route, but they only take one step of those routes if they are in the same cell as $a_4$ and $a_4$ is in a state indicating that $a_2$, resp.\ $a_3$ should move (the latter condition is not strictly necessary, but simplifies things by ensuring that $a_2$ and $a_3$ never move at the same time).
Agent $a_4$ alternates between visiting $a_2$ and $a_3$, during each ``visit'' making sure that the respective agent takes the desired number of steps (i.e., $k-1$ or $k+1$).
It does so by going to the cell of the respective agent $a_i$ ($i \in \{2, 3\}$), indicating that $a_i$ should take a step of its route, waiting until $a_i$ takes a step and leaves the cell, incrementing an internal counter, following agent $a_i$ to the next cell, and repeating this behavior until the counter indicates that the desired number of steps has been taken by $a_i$, upon which $a_4$ visits the other agent $a_{5-i}$.
Note that $a_4$ always knows in which direction it has to move to find the desired agent as the coordinates of $a_2$ and $a_3$ only differ in dimension $1$, and $a_2$ always has a smaller (or equally large) first coordinate.
Moreover, $a_4$ also knows in which direction it has to go to follow the agents to the next cell as the only change in direction is performed by $a_2$ and the reason for the change, namely meeting $a_1$, is an information known to $a_4$ since when $a_2$ meets $a_1$, it stays in the cell containing $a_1$ until $a_4$ also arrives there.

In an analogous fashion, we can implement \divstack{k} with $4$ semi-synchronous agents.
For the other four subroutines, the picture is even simpler: it is straightforward to check that these subroutines can already by implemented by $3$ semi-synchronous agents by having the agents perform the same steps as in the respective synchronous subroutines.
The reason that these subroutines also work in the semi-synchronous setting is that either the synchronous version already contain one agent that effectively acts as a synchronizer (in the sense that every action is performed by that agent or directly instigated by a visit of that agent), as in \isdiv{k} and \movestack{g, i}, or the actions of the agents are independent of each other, as in \initstack{k} and \incstack{k}.
For these four subroutines, we will simply assume that $a_4$ is treated the same as $a_3$; in particular, at the beginning and end of each subroutine, we will always have $a_2$, $a_3$, and $a_4$ in the same cell, indicating the end of the stack.

We have to be a bit careful with the termination of each subroutine as we want to be able to concatenate the subroutines without problems.
To this end, we will again use $a_4$ as a synchronizer: before terminating itself, $a_4$ will wait that $a_2$ and $a_3$ (which are in the same cell at the end of each subroutine) have terminated.
Similarly, we can assume that $a_4$ will initialize the next subroutine by changing its state suitably, thereby making sure that the start and end of the subroutines align across all agents.
A last detail is that in the semi-synchronous version of \movestack{g, i} (which is the only subroutine where $a_1$ moves), after meeting $a_1$, agent $a_2$ has to wait until $a_1$ takes its step before moving back to $a_3$ and $a_4$, in order to make sure that $a_4$ does not terminate and initialize the next subroutine before $a_1$ takes its step.

\section{The Exploration Protocol}\label{sec:protocol}

In this section, we will combine the building blocks of Section~\ref{sec:blocks} to a protocol that allows $3$ synchronous, resp.\ $4$ semi-synchronous, agents to explore the $n$-dimensional (oriented) grid, and prove the protocol's viability.
Our protocol is given by algorithm \exsync{}.


\begin{algorithm}
	\caption{\exsync}
	\label{alg:exsync}
	\begin{algorithmic}[1]
		\State \initstack{3}
		\StartRepeat
		\For{each function $g:\{1, \dots, n\} \to \{-1, 1\}$}
		\State \followroute{g}
		\State \followroute{-g}
		\EndFor
		\State \incstack{2}
		\EndRepeat
		\item[]
		\Procedure{FollowRoute}{$g$}
		\For{$i = 1$ to $n$}
		\While{\isdiv{p_i}}
		\State \divstack{p_i}
		\State \multstack{2}
		\State \movestack{g(i), i}
		\EndWhile
		\While{\isdiv{2}}
		\State \divstack{2}
		\State \multstack{p_i}
		\EndWhile
		\EndFor
		\EndProcedure
	\end{algorithmic}
\end{algorithm}

The underlying idea of algorithm \exsync{} is the same as in the algorithms from \cite{Dobrev2019}:
We generate each (non-zero) $n$-dimensional vector $(v_1, \dots, v_n)$ with non-negative integer coordinates, and for each such vector, we let one agent walk from the origin to each cell $(c_1, \dots, c_n)$ such that $c_i \in \{v_i, -v_i\}$ for all $1 \leq i \leq n$, and then back to the origin.
More precisely, in each execution of \followroute{g}, agent $a_1$ walks to the respectively specified cell $(c_1, \dots, c_n)$, and in each execution of \followroute{-g}, $a_1$ walks back to the origin.
To generate $(v_1, \dots, v_n)$, a counter, represented by the stack size, is used that is incremented gradually, thereby iterating through the positive integers.
Each time the counter is incremented, the new value $X$ will be transformed into some $n$-dimensional vector, $(v_1, \dots, v_n)$, where the design of the transformation has to make sure that every (non-zero) vector with non-negative integer coordinates is generated by some value $X$.
For technical reasons we require the stack size to be odd, so whenever we increase the counter, we will increase it by $2$, while still ensuring that all vectors are generated.

However, as we have one fewer agent available than in the protocols in \cite{Dobrev2019}, our protocol requires a new way to implement this idea.
In particular, we avoid using one separate agent to remember the stack size when the stack is read, instead making sure that even after the stack is read, no information about the previous stack size(s) is lost\footnote{Note that such information is still required after reading the stack: we will need it both to guide $a_1$ back to the origin and to retrieve the counter value $X$ that we want to increase repeatedly.}.
To this end, we define the vector $(v_1, \dots, v_n)$ we want to transform $X$ into as follows.
Let $p_1, \dots, p_n$ denote the first $n$ odd primes, where $p_1 < \dots < p_n$.
For all $1 \leq i \leq n$, we define $v_i$ to be the largest non-negative integer such that $p_i^{v_i}$ divides $X$.
In other words, $v_i$ represents how often $p_i$ occurs as a prime factor of $X$.

Consider procedure \followroute{g}.
The for loop of this procedure iterates through the $n$ dimensions.
For each dimension $i$, the first while loop repeatedly replaces one prime factor $p_i$ by prime factor $2$, by dividing by $p_i$ and multiplying by $2$.
Each time such a replacement is performed, the whole stack is moved one cell w.r.t.\ dimension $i$ (either increasing or decreasing the respective coordinate by $1$, depending on the value of $g(i)$).
After all (i.e., $v_i$) occurrences of $p_i$ as prime factors have been replaced by factors $2$, the stack manipulations are reversed in the second while loop, resulting in the original stack size $X$.
Note that in the very beginning of algorithm \exsync{}, the stack size is initialized to $3$, and each time the counter represented by the stack size is increased, it is increased by $2$; hence, before starting the first while loop, the stack size is odd, ensuring that the second while loop goes through exactly the same number of iterations as the first one.
Note further that we do not revert the steps that $a_1$ took (yet) when reversing the stack manipulations.
After iterating through all dimensions, agent $a_1$ is now in cell $(c_1, \dots, c_n)$, and we can consider this cell as explored, concluding the execution of \followroute{g}.

The execution of \followroute{-g} is identical to the execution of \followroute{g}, except that each step of $a_1$ is performed in the opposite direction.
Hence, at the end of the execution of \followroute{-g}, agent $a_1$ is back at the origin, while the stack size is (again) $X$.
The (outer) for loop in algorithm \exsync{} simply iterates through all possible assignments of signs $\in \{ -1, +1 \}$ to the dimensions, making sure that for each generated vector $(v_1, \dots v_n)$, each corresponding cell $(v'_1, \dots, v'_n)$ is explored.


\oriented*
\begin{proof}
  First, we observe that the design of \exsync{} ensures that each grid cell will be explored at some point in time since each (non-zero) cell $(c_1, \dots, c_n)$ is explored in the iteration where the counter value is $\prod_{i=1}^n p_i^{|c_i|}$.
	Note that each such counter value is odd (and at least $3$) and thus indeed reached by starting with a counter value of $3$ and repeatedly increasing it by $2$.
	Hence to prove the theorem, it suffices to show that \exsync{} can be executed with $3$ synchronous finite automata, resp.\  $4$ semi-synchronous finite automata, which we do in the following.

	For any fixed $n$, there is only a finite number of \emph{distinct} subroutines that we execute in \exsync{}.
	Each such subroutine can be executed by $3$ synchronous agents, resp.\  $4$ semi-synchronous agents, due to the implementations provided in Section~\ref{sec:blocks}.
	Moreover, agents $a_2$ and $a_3$ (or, in the case of $4$ semi-synchronous agents, agent $a_4$) know when a subroutine has terminated and can start the next one, while agent $a_1$ only moves or changes its state in the subroutines upon a visit by another agent, and finishes its actions before the termination of the subroutine.
	Hence, the theorem follows if we can show that the choice of each subsequent subroutine by $a_2$ and $a_3$, resp.\ $a_4$, can be implemented using a finite automaton.
	However, this follows from the following observations.

	If the currently executed subroutine is not \isdiv{p_i} or \isdiv{2}, then the subroutine executed next according to \exsync{} is uniquely defined by the tuple consisting of the current subroutine (type), the current function $g$, the current value of $i$, the information whether we are in procedure \followroute{g} or in procedure \followroute{-g} (or neither), and the information whether we are in the first or the second while loop of the respective procedure.
	If the current subroutine is \isdiv{p_i} or \isdiv{2}, then, in order to determine the next subroutine, we additionally need the answer to the question whether the current stack size is divisible by $p_i$ or $2$, respectively.
	As there are only finitely many such tuples and the answer to the divisibility question adds just one bit of information, determining each subsequent subroutine can be performed by a finite automaton.
\end{proof}

\section{Unoriented Grids}\label{sec:unogrids}
In \cite{Dobrev2019}, the authors showed that any protocol for the oriented grid can be transformed into a protocol for unoriented grids by adding sufficiently many agents such that, at all times, each original agent moving across a non-constant distance is accompanied by one of the additional agents.
In particular, for both their protocol and our improved protocol, this implies that $2$ additional agents are required in the synchronous case and $1$ additional agent in the semi-synchronous case (since in the protocols for the oriented grid, $2$ synchronous agents are traversing non-constant distances at the same time, while in the semi-synchronous case only $1$ agent does so).
Hence, our protocol for the oriented grid improves also the state of the art for the minimum number of required agents on \emph{unoriented} grids from $6$ to $5$ (in both the synchronous and the semi-synchronous setting).

On an informal level, it seems unlikely that our upper bound of $5$ can be improved since intuitively, as Dobrev et al.~\cite[Section 7]{Dobrevarxiv} write, ``a lone agent cannot cross any non-constant distance, as the irregular nature of the port labels would lead it astray, never to meet any other agent''.
The tightness of our bound on the oriented grid combined with the perceived necessity of having moving agents accompanied by a partner seems to indicate that we cannot do better.
However, there is no formal proof of any lower bound beyond the synchronous $3$-agent and semi-synchronous $4$-agent lower bounds~\cite{Brandt2018, Emek15} that carry over from the case of the oriented grid.
Admittedly, as such a formal lower bound might require us to find a ``bad'' input instance (i.e., a bad input edge labelings of the infinite $n$-dimensional grid) for every potential protocol with more than $3$ synchronous, resp.\ $4$ semi-synchronous, agents, it is not particularly surprising that we do not have better lower bounds---yet, making at least some progress would be desirable.

In this section, we will show that under a natural assumption the current lower bounds are actually optimal by providing tight upper bounds.
Our assumption states that you cannot walk in a cycle if you always follow the same direction, or, more formally:
\begin{assumption}\label{ass:cycle}
	Let $\ell \in \{ 1, \dots, 2n \}$ be any port, $z$ any positive integer, and $c_0, \dots, c_z$ any sequence of cells such that, for each $0 \leq j \leq z-1$, we reach cell $c_{j+1}$ by leaving cell $c_j$ via port $\ell$.
	Then $c_0 \neq c_z$.
\end{assumption}
Surprisingly, this assumption does not contradict the intuition about agents traveling alone discussed above, yet it still allows us to prove upper bounds for unoriented grids matching the lower bounds obtained on \emph{oriented} grids.
While our upper bounds answer the question for the minimally required number of agents in a natural\footnote{After all, it seems like a reasonable minimal requirement for a sense of direction that if you go north (or in any other direction) repeatedly, then you do not return to the starting point.} setting very close to truly unoriented grids, we think that they also constitute a useful step on the way to a lower bound construction for the general unoriented setting (assuming that the current lower bounds are not optimal): any such construction necessarily has to contain cycles that violate Assumption~\ref{ass:cycle}.

\paragraph{Our Approach.}
The general idea behind our approach is to find a way to construct a stack also for unoriented grids.
The natural idea of simply selecting one port $\ell$ and interpreting the sequence of cells obtained by successively leaving cells via port $\ell$ as the stack does not work: while it is easy for an agent to traverse the stack in the direction away from the base (it just has to leave each cell via port $\ell$), traversing the stack in the opposite direction runs into the problem that there might be different neighboring cells from which the current cell can be reached via port $\ell$ and the traversing agent cannot know which is the one that belongs to the intended stack.
Instead, we will build the desired virtual stack by constructing an auxiliary (infinite) directed labeled forest and then traversing (a part of) the forest from some starting cell in a DFS-like fashion, which will ensure that agents can traverse the stack in both directions.
In particular, the same cell can occur in the stack several times; to distinguish the occurrences (and make it possible for an agent to traverse the virtual stack), our stack will formally consist of pairs (cell, integer), where the integers come from the set $\{1, \dots, 2n\}$.
For an illustration of the auxiliary graph and the virtual stack, we refer to Figures~\ref{fig:grid} and~\ref{fig:tree}.

%
%
%

\begin{figure}
	\centering
	\includegraphics[scale=1.7]{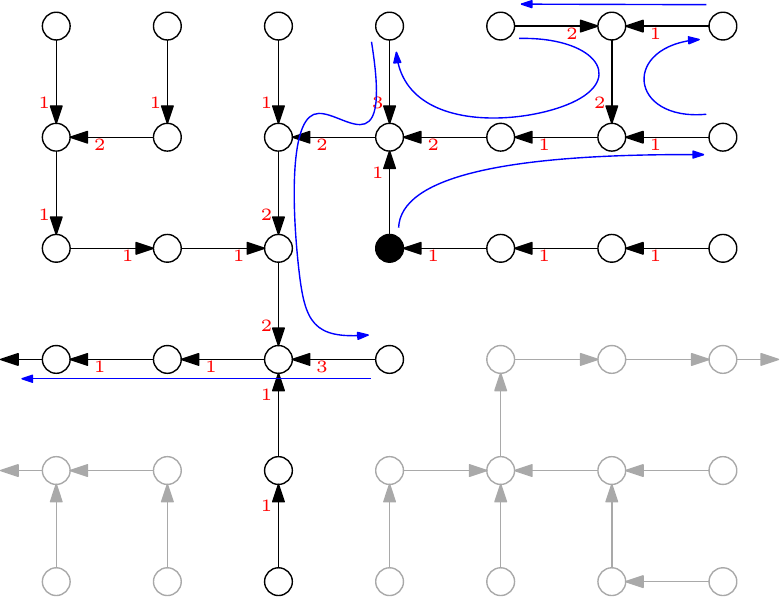}
	\caption{Figure~\ref{fig:grid} depicts a part of a possible auxiliary graph $G$ for a $2$-dimensional unoriented grid and the respective virtual stack. Vertices, i.e., cells, are represented by circles, and directed edges by arrows. The parts grayed out belong to different trees than the one containing the cell where $a_1$ is located (colored black). The edges are labeled with their respective levels. The physical cells of the virtual stack rooted in the black cell (i.e., the first component of the pairs the virtual stack consist of) are indicated by the route that starts in the black cell and follows the blue arrows. Each further step on this route leads to the physical cell corresponding to the next higher position in the stack, where the black cell indicates position $0$. For each occurrence of a cell on this route, the corresponding level (i.e., the second component of the pairs the virtual stack consists of) is $1$ if the current cell is a child of the previous cell, and equal to the level of the arrow traversed last plus $1$ if the current cell is the parent of the previous cell. When the blue route goes from a cell to its parent, then the edge traversed next will have a level that is higher by $1$ than the previously traversed edge; when the route goes from a cell to one of its children, then the edge traversed next has level $1$. This leads to a virtual stack that corresponds to a part of a DFS exploration on the (infinite) tree containing the black cell, as can be seen in Figure~\ref{fig:tree}, where the same route on the same auxiliary graph is depicted as a rooted tree.}\label{fig:grid}
\end{figure}

\begin{figure}[t]
	\centering
	\includegraphics[scale=1.35]{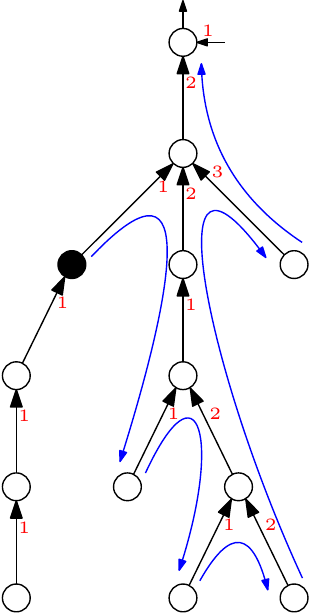}
	\caption{The same virtual stack as in Figure~\ref{fig:grid}, depicted as a rooted tree.}\label{fig:tree}
\end{figure}

\paragraph{The Auxiliary Graph.}
We start by defining our auxiliary graph $G=(V, E)$.
The vertices of $G$ are the cells of our grid, and we have a directed edge $(c, c')$ between two cells $c, c'$ if $c$ and $c'$ are neighbors in the grid and cell $c'$ is reached by leaving cell $c$ via port $1$.\footnote{The choice of port $1$ here is arbitrary; choosing any other label from $\{1, \dots, 2n\}$ works equally well. }
In particular, this implies that each cell $c$ has exactly one outgoing edge in $G$; we call the cell $c'$ reached by traversing this edge the \emph{parent of $c$}, and $c$ a \emph{child of $c'$}.
Note that Assumption~\ref{ass:cycle} ensures that $G$ does not contain cycles, and hence, is an infinite forest.
In particular, for any two neighboring cells $c, c'$, at most one of the two possible edges $(c, c')$ and $(c', c)$ is present in $E$.

Let $\indeg{c}$ denote the \emph{indegree} of a cell $c$, i.e., the number of edges from $E$ incoming to $c$.
We assign to each edge $e = (c, c')$ a level $L(e)$ as follows.
For each cell $c'$, we order the incoming edges $(c, c')$ increasingly by the corresponding port of $c'$, and then assign (distinct) \emph{levels} from $1$ to $\indeg{c'}$ to the edges according to this order.
For instance, if $c'$ has two incoming edges $(c, c')$ and $(c'', c')$, corresponding to ports $5$ and $3$ of $c'$, respectively, then the order of the edges will be $(c'', c')$, $(c, c')$, and we will assign level $1$ to $(c'', c')$, and level $2 = \indeg{c'}$ to $(c, c')$.

\paragraph{The Virtual Stack.}
Using the auxiliary graph $G$, we now define, for each cell $c$, the \emph{virtual stack $\virt_c$ rooted in $c$} as follows.
Recall that $\virt_c$ consists of pairs (cell, integer).
We will use the functions $\cell{\cdot}$ and $\level{\cdot}$ to retrieve the first, resp.\ second, component of such a pair.
The base of the stack is defined as $\virt_c[0] := (c, \indeg{c} + 1)$.
For each integer $j \geq 1$, we inductively define $\virt_c[j]$ according to the following case distinction.
\begin{itemize}
	\item If $\level{\virt_c[j-1]} = \indeg{\cell{\virt_c[j-1]}} + 1$, then	
	\begin{itemize}
		\item $\cell{\virt_c[j]}$ is defined as the parent of $\cell{\virt_c[j-1]}$, and
		\item $\level{\virt_c[j]} := L((\cell{\virt_c[j-1]}, \cell{\virt_c[j]})) + 1$.
	\end{itemize}
	\item If $\level{\virt_c[j-1]} \leq \indeg{\cell{\virt_c[j-1]}}$, then
	\begin{itemize}
		\item $\cell{\virt_c[j]}$ is defined as the child of $\cell{\virt_c[j-1]}$ that is connected to \newline $\cell{\virt_c[j-1]}$ via an (outgoing) edge of level $\level{\virt_c[j-1]}$, and
		\item $\level{\virt_c[j]} := 1$.
	\end{itemize}
\end{itemize}
In other words, we inductively build the virtual stack rooted in $c$ as follows.
We start in $c$ and leave $c$ via the unique outgoing edge.
Each time we enter a cell $c'$ via an incoming edge, i.e., coming from a child $c''$, the next cell we visit is the next higher child of $c'$, i.e., the child that is connected to $c'$ via an edge of level $L(c'', c') + 1$.
If no higher child remains, i.e., if $(c'', c')$ has level $\indeg{c'}$, then the next cell we visit is the parent of $c'$.
Each time we enter a cell $c'$ from its parent, the next cell we visit is the first child of $c'$, i.e., the child that is connected to $c'$ via an edge of level $1$.
Hence, our stack corresponds to a DFS exploration of the tree in forest $G$ containing $c$, where we assume that the part of the DFS that is executed before traversing the edge from $c$ to its parent has already happened.
As the tree is infinite, we may not reach every cell contained in the tree in finite time, but to use such a DFS exploration as a stack, this is not relevant.
What is relevant, however, is that once we traverse an edge from a child to its parent, the DFS will never return to the child in finite time as Assumption~\ref{ass:cycle} ensures that the parent chain starting from $c$ (and therefore also any parent chain starting from any other cell visited by the partial DFS) is infinite.
Combining this fact with the cyclic fashion in which each visited cell iterates through its children and parent to determine the neighbor visited next, we obtain the following observation.

\begin{observation}\label{obs:norepeat}
	Fix an arbitrary cell $c$.
	For any two non-negative integers $i \neq j$, we have $\virt_c[i] \neq \virt_c[j]$.
\end{observation}

In order to make use of the defined virtual stack, we need the agents to be able to represent their position in the stack in some way.
However, given the specific design of the virtual stack, this is not difficult: each agent allocates a part of its state to keep track of the level $\level{\virt_c[j]}$ of the current position $\virt_c[j]$ in the stack, while the first component $\cell{\virt_c[j]}$ of the current position in the stack is simply represented by the cell the agent currently occupies.
An advantage of this design is that each agent $a_i$ can determine which other agents are in the same stack position as $a_i$, and which are not (despite possibly being in the same physical cell).
In other words, each agent has all the necessary information to evaluate its transition function, even for moving on the virtual stack.

However, there is one piece still missing for using the virtual stack similar to a physical stack: we have to show that even a lone agent can traverse the virtual stack in either direction, i.e., that a finite automaton is sufficient to determine the physical cell that corresponds to the previous, resp.\ subsequent, position in the virtual stack, and similarly, to determine the level of that stack position.
The following lemma takes care of this.

\begin{lemma}\label{lem:virtualmove}
	There is a finite automaton that, when located in cell $\cell{\virt_c[j]}$ in state $(i, \level{\virt_c[j]})$, where $c$ is an arbitrary cell, $i \in \{-1, 1\}$, and $j \geq 1$ an arbitrary integer, moves to cell $\cell{\virt_c[j+i]}$ and changes its state to $(i, \level{\virt_c[j+i]})$ in $2$ time steps.
\end{lemma}
\begin{proof}
	Let $c' = \cell{\virt_c[j]}$ denote the cell in which the finite automaton is located.
	By the definition of our auxiliary graph $G$, the input labels on the grid edges incident on $c'$ (which are part of the information available to the finite automaton) uniquely determine the edges incident to $c'$ in $G$, their orientations (hence, also $\indeg{c'}$), and, if they are incoming edges, their levels.
	Denote the entirety of this uniquely determined information about the edges incident on $c'$ by $\mathcal I_{c'}$.
	By the definition of $\virt_c$, which of the incident edges leads to $\cell{\virt_c[j+1]}$ is uniquely determined by $\level{\virt_c[j]}$ and $\mathcal I_{c'}$.
	Moreover, $\level{\virt_c[j+1]}$ is uniquely determined by $\level{\virt_c[j]}$, $\mathcal I_{c'}$, and $\mathcal I_{c''}$, where $c'' = \cell{\virt_c[j+1]}$.
	Hence, if $i = 1$, then there is a finite automaton that first moves to $\cell{\virt_c[j+1]}$ and then, (potentially) using input information obtained in the new cell, changes its state to $(i, \level{\virt_c[j+i]})$.

	For the case $i = -1$, we observe that, according to the definition of $\virt_c$, the following hold.
	\begin{itemize}
		\item If $\level{\virt_c[j]} = 1$, then
		\begin{itemize}
			\item $\cell{\virt_c[j-1]}$ is the parent of $\cell{\virt_c[j]}$, and
			\item $\level{\virt_c[j-1]} = L(\cell{\virt_c[j]}, \cell{\virt_c[j-1]})$.
		\end{itemize}
		\item If $\level{\virt_c[j]} > 1$, then
		\begin{itemize}
			\item $\cell{\virt_c[j-1]}$ is the child of $\cell{\virt_c[j]}$ that is connected to $\cell{\virt_c[j]}$ via an edge of level $\level{\virt_c[j]} - 1$, and
			\item $\level{\virt_c[j-1]} = \indeg{\cell{\virt_c[j-1]}} + 1$.
		\end{itemize}
	\end{itemize}
	It follows that, similarly to before, which of the edges that are incident to $\cell{\virt_c[j]}$ leads to $\cell{\virt_c[j-1]}$ is uniquely determined by $\level{\virt_c[j]}$ and $\mathcal I_{c'}$, and $\level{\virt_c[j-1]}$ is uniquely determined by $\level{\virt_c[j]}$, $\mathcal I_{c'}$, and $\mathcal I_{c'''}$, where $c''' = \cell{\virt_c[j-1]}$.
	Now we obtain the lemma analogously to the case $i = 1$.
\end{proof}

Note that when applying Lemma~\ref{lem:virtualmove}, the finite automaton from Lemma~\ref{lem:virtualmove} will only constitute a part of the finite automaton governing our agents in the final protocol for the unoriented case.

In order to explore unoriented grids, we will rely heavily on our protocol for oriented grids.
While the use of a virtual stack will take care of some of the difficulties resulting from the missing global consistency of the provided input edge labeling, there is an additional obstacle for using \exsync{} to explore unoriented grids:
In \exsync{}, each cell is reached by a route that starts in the origin, proceeds along dimension $1$, then proceeds along dimension $2$, and so on, until we have exhausted all dimensions.
On unoriented grids, the natural analogon would be to use the \emph{local} orientations available to us in the same manner, i.e., starting in the origin we first leave cells via port $1$, then we switch to port $2$, and so on.
However, due to the fact that there is no global consistency guarantee for those local orientations, it might be that some cells cannot be reached via such a route from the origin.
To overcome this obstacle we will make sure that $a_1$ does not iterate through all possible \emph{ports} in the manner described above, but instead through all \emph{globally consistent dimensions} as in the oriented case, i.e., for the route that takes $a_1$ to the cell to be explored next, the input edge labels will be essentially irrelevant.

The underlying idea to make this approach work in the unoriented setting is to use the \emph{handrail technique}, developed for $2$-dimensional grids by Mans \cite{mans1997optimal} and generalized to higher dimensions in \cite{Dobrevarxiv}.
This technique ensures the following:
Assume that agent $a_1$ in some cell $c'$ is aware of a bijection $f$ from the set $\{1, \dots, 2n\}$ of the ports of $c'$ to the set $\{-1, +1\} \times \{1, \dots, n\}$ such that, for any $i \in \{1, \dots, n\}$, the two ports that are mapped to $(-1, i)$ and $(+1, i)$ lead to opposite neighbors of $c'$, i.e., to two cells that only differ in one coordinate (by $2$). Then, if there is another agent $a'$ in the same cell $c'$, there is a $2$-agent protocol (both in the synchronous and in the semi-synchronous setting) that moves both agents to any adjacent cell $c''$ of their choice and provides $a_1$ with a new bijection $f'$ with the same properties as $f$ such that the following holds for any $j \in \{-1, +1\}$ and any $i \in \{1, \dots, n\}$: leaving $c'$ via the port that $f$ maps to $(j, i)$ leads in the same global direction (i.e., along the same dimension with the same sign) as leaving $c''$ via the port that $f'$ maps to $(j, i)$.
In other words, once agent $a_1$ is aware of the different dimensions and has chosen a name for each dimension and an orientation of each dimension, it can preserve this knowledge (in a globally consistent manner) as long as another agent is available each time $a_1$ moves to an adjacent cell.

Moreover, the handrail technique also allows agent $a_1$ to generate an initial bijection $f$ as described above with the help of a second agent.
Note that agent $a_1$ has no knowledge of the actual order of the dimensions specified by the global coordinates of the cells, i.e., the direction $a_1$ understands as, e.g.,  $(-1, 5)$ (via its bijection) might be $(+1,2)$ in truth.
Hence, there is no guarantee in which order \exsync{} iterates through the true global dimensions if \exsync{} uses the dimension names that $a_1$ maintains instead of the dimension names that the global coordinates specify.
However, in order to reach every cell, the \emph{order} in which \exsync{} iterates through the dimensions is irrelevant as long as \emph{the same order} is used for each route that $a_1$ takes, and the latter is guaranteed if each time $a_1$ moves to a new cell, there is another agent that helps $a_1$ maintain its knowledge during the move. 

Now we are set to prove our main result for unoriented grids.



\newcounter{tempcounter}
\setcounter{tempcounter}{\value{theorem}}
\setcounter{theorem}{\value{unoricounter}}
\newcounter{tempsection}
\setcounter{tempsection}{\value{section}}
\setcounter{section}{\value{unorisection}}

\begin{theorem}
	Suppose that Assumption~\ref{ass:cycle} holds.
	Then, for any positive integer $n$, $3$ synchronous finite automata, resp.\ $4$ semi-synchronous finite automata, suffice to explore any $n$-dimensional unoriented grid.
\end{theorem}
\setcounter{theorem}{\value{tempcounter}}
\setcounter{section}{\value{tempsection}}
\begin{proof}
	As in the proof of Theorem~\ref{thm:oriented}, we will let the agents execute algorithm \exsync{} to explore the given grid.
	However, there are two essential details.
	First, the stack that \exsync{} refers to will now be the virtual stack defined in this section.
	More precisely, for each executed subroutine, the stack used in this subroutine is the virtual stack $\virt_c$ rooted in $c$, where $c$ is the cell in which $a_1$ is located in the very beginning of the subroutine.
	Second, and in contrast, the directions that $a_1$ uses in its exploration will be given by the globally consistent dimensions and orientations of the dimensions.
	By our discussions regarding the handrail technique, $a_1$ can choose arbitrary (but fixed) names for the dimensions and an orientation for each dimension in the beginning, as initially all agents are in the origin.
	Moreover, as we will argue in the very end of the proof, each time $a_1$ moves, another agent will be available to help $a_1$ maintain the chosen globally consistent orientation.
	As the design of \exsync{} ensures that every cell will be explored, what remains to be shown, apart from the availability of a helper agent when $a_1$ moves, is that each subroutine can be performed by $3$ synchronous, resp.\ $4$ semi-synchronous agents, also on unoriented grids.

	In the synchronous case, this follows for \multstack{k}, \divstack{k}, \isdiv{k}, \initstack{k}, and \incstack{k} by Observation~\ref{obs:norepeat} and Lemma~\ref{lem:virtualmove}, which ensure that the virtual stack behaves like a physical stack and can be traversed by any agent, even if alone, in both directions.
	Note that it is not a problem that Lemma~\ref{lem:virtualmove} only guarantees that a move on the virtual stack can be executed in $2$ time steps (and not necessarily $1$)---on unoriented grids we simply allocate two time steps for every time step of the algorithm on oriented grids, and if an agent does not need the second time step to complete its action, it simply waits for one time step.

	The only subroutine in the synchronous setting requiring special care is \movestack{g,i} as this is the only subroutine where $a_1$ moves, and hence, the only subroutine where the cell in which our virtual stack is rooted changes.
	The main issue making the execution of this subroutine somewhat difficult is that in unoriented grids following different directions (specified by ports) is not commutative anymore: leaving a cell $c'$ via some port $\ell$ and the reached cell via some port $\ell'$ does not necessarily lead to the same cell as if we leave $c'$ via port $\ell'$ and the reached cell via port $\ell$.
	In contrast, commutativity of directions holds on oriented grids, and we make use of this fact in \movestack{g,i} by moving the endpoints of the stack one cell in the same direction and being assured that this does not change the fact that we can reach the end of the stack from the base by repeatedly walking north.

	We overcome this obstacle by executing \movestack{g,i} in unoriented grids in a completely different way:
	Let $c'$ denote the root of the virtual stack before executing \movestack{g,i}, and $c''$ the root of the virtual stack after executing \movestack{g,i}, i.e., the virtual stack changes from $\virt_{c'}$ to $\virt_{c''}$.
	For agent $a_1$ nothing changes---it simply moves from $c'$ to $c''$ as it would on oriented grids.
	Agents $a_2$ and $a_3$ first traverse $\virt_{c'}$ until they reach $c'$, then move from $c'$ to $c''$, and finally traverse $\virt_{c''}$ until the new virtual stack $\virt_{c''}$ has the same size as $\virt_{c'}$ had at the start of \movestack{g,i}.
	In order to ensure that the new stack has the same size as the previous one, $a_2$ travels with speed $1$ when traversing $\virt_{c'}$ and with speed $1/2$ when traversing $\virt_{c''}$ whereas $a_3$ travels with speed $1/2$ when traversing $\virt_{c'}$ and with speed $1$ when traversing $\virt_{c''}$; when they meet again on $\virt_{c''}$ (in aligned states, i.e., immediately before both would take another step), they must have traversed the same number of (virtual stack) cells on both stacks.
	Again, executing this behavior on virtual stacks is possible by Observation~\ref{obs:norepeat} and Lemma~\ref{lem:virtualmove}.

	In the semi-synchronous setting, the theorem basically follows by using the semi-synchronous implementations of the subroutines (as presented in Section~\ref{sec:blocks}) and observing that the new synchronous implementation for \movestack{g,i} can be transformed into a semi-synchronous implementation by using the additional agent $a_4$ as a synchronizer (analogous to how this is done for, e.g., subroutine \multstack{k} in Section~\ref{sec:blocks}).
	The only change compared to the case of oriented grids (beyond the new implementation of \movestack{g,i}, and requiring Observation~\ref{obs:norepeat} and Lemma~\ref{lem:virtualmove} to guarantee that virtual stacks can be used in the same way as physical stacks) is that each movement on the virtual stack may take $2$ time steps according to Lemma~\ref{lem:virtualmove}.
	The only effect this has on the viability of our subroutine implementations is that an agent acting as a synchronizer cannot be sure that some other agent that leaves the synchronizer's cell has completed its move on the virtual stack before the synchronizer is scheduled again (and possibly prompts some other agent to move).
	However, this is easily remedied: whenever the synchronizer wants some other agent to take a step on the virtual stack, it simply follows that agent to the respective cell and waits until the agent is in the correct state indicating that the step on the virtual stack is completed.

	Finally, we take care of the requirement that, each time agent $a_1$ moves, a helper agent must be available.
	As the only subroutine in which $a_1$ moves is \movestack{g,i}, and the new implementation of \movestack{g,i} specifies that each agent first moves to the cell $c'$ containing $a_1$, then to the adjacent destination cell $c''$ of $a_1$, and only then onwards, such a helper agent is available:

	In the synchronous case we require $a_3$ to be the helper agent as it arrives later than $a_2$ at cell $c'$ (note that as soon as $a_2$ arrives, $a_1$ is aware of the port via which it reaches $c''$ and can inform $a_2$ about the direction; the helper agent is only needed to let $a_1$ preserve its knowledge of the dimensions during its move).
	We remark that we need to make sure that $a_2$ and $a_3$ spend the same amount of time between finishing the traversal of stack $\virt_{c'}$ and starting the traversal of stack $\virt_{c''}$ to make sure that the sizes of the two stacks will be the same.
	However, this is not difficult: there is a finite upper bound $T_n$ (depending on the dimension $n$ of the grid) for the time it takes to perform the $2$-agent protocol that moves $a_1$ (and $a_3$) from $c'$ to $c''$ and preserves $a_1$'s knowledge; hence both $a_2$ and $a_3$ can simply start the traversal of stack $\virt_{c''}$ $T_n$ time steps after their respective completion of the traversal of stack $\virt_{c'}$.
	In the semi-synchronous case, we can simply use our synchronizer $a_4$ as the helper agent, at a point in time when both $a_2$ and $a_3$ have already completed their traversals of $\virt_{c'}$ and reached $c''$.

	There is a small detail that we omitted in the discussion of the semi-synchronous implementation of \movestack{g,i} that we now take care of: when $a_2$ has already started its traversal of $\virt_{c''}$, but $a_3$ has not finished its traversal of $\virt_{c'}$ yet in the new implementation of \movestack{g,i}, our synchronizer $a_4$ will have to be able to move between those two agents despite the fact that you cannot necessarily reach one from the other by simply leaving cells always via the same port.
	When traveling from $a_3$ to $a_2$, this can be achieved straightforwardly: $a_4$ knows how to travel on either stack, and it also knows when to perform the one step from $c'$ to $c''$ since $a_4$ can detect when it reached the base of stack $\virt_{c'}$ by noticing the presence of $a_1$ in the physical cell $c'$ and checking that it reached $c'$ via an edge that is outgoing from $c'$ in the auxiliary graph $G$.
	The information that $a_4$ is missing in order to do the same when traveling from $a_2$ to $a_3$ is whether or not a reached (physical) cell is $c''$, since $c''$ is not marked by the presence of $a_1$ as $a_1$ has not performed its move yet.
	However, this can be remedied by letting $a_4$ explore all adjacent cells each time it reaches a new cell $c'''$ after taking a step on $\virt_{c''}$; if none of them contains $a_1$, $c''' \neq c''$, and if one of them contains $a_1$, agent $a_4$ can infer whether $c''' = c''$ by asking $a_1$ about the port that leads from $c'$ to $c''$.	
\end{proof}

\setcounter{theorem}{\value{tempcounter}}

\section{Polynomial-Time Exploration}\label{sec:poly}

All algorithms we considered so far require a super-polynomial number of steps to reach a treasure located at distance $D$ from the origin.
In fact, even to encode the coordinates of a cell at distance $D$, we need a stack of size exponential in $D$.
To speed up the exploration, Dobrev et al.~\cite{Dobrev2019} changed the exploration strategy to the following:
They explore a hypercube of side length $h$ centered at the origin (for increasing values of $h$).
For this, they change their stack implementation, that was based on a binary encoding of $n$-tuples, to an implementation based on an encoding of $n$-tuples using a $h$-ary alphabet, which captures exactly the cells in the hypercube.
To perform the exploration of those cells, they require an additional operation, which is to multiply the stack size by $h$, for any (non-constant) value of $h$.
They show that this extra operation can be implemented using one additional agent over their exponential time deterministic protocols, in both the synchronous and semi-synchronous setting.
This yields a protocol that uses 5 agents in the synchronous, and 6 agents in the semi-synchronous setting, respectively.

We will show how to adapt our results from the previous sections to perform the operations \multstack{h}, \isdiv{h}, and \divstack{h} on oriented grids for non-constant values of $h$.
Our implementations will only be applicable for $h = 2^i$ being a power of 2, however this is sufficient for our purpose.
To perform these operations, we will need one extra agent, mainly to encode the value of $h$.
In fact, we will show that we can use multiple calls to \multstack{k} for constant values of $k$, each time interpreting different agents as encoding a stack.

Once we have these stack operations for non-constant values of $h$, we can use the result in a black-box fashion with techniques from Dobrev et al. \cite{Dobrev2019}, to get a protocol with polynomial exploration time using 4 agents in the synchronous, and 5 agents in the semi-synchronous setting.

\subparagraph{\multstack{h}}
We will focus on this operation, the adaptations for the other operations are analogous.
Let us also focus on the synchronous setting, in the semi-synchronous setting we can use one additional agent which acts as a synchronizer, similar to before.
We will overload notation, and use \multstack{k, a_1, a_2, a_3}, \divstack{k, a_1, a_2, a_3}, and \isdiv{k, a_1, a_2, a_3} to indicate that we manipulate the stack formed by agents $a_1, a_2$, and $a_3$, where $a_1$ forms the base and $a_2$ and $a_3$ are co-located at distance $X$ from $a_1$.
For simplicity, we will also denote the base of the stack by $a_1$.

\begin{algorithm}
	\caption{\multstack{h}}
	\label{alg:multh}
	\begin{algorithmic}[1]
		\item[\textbf{Input:} $h = 2^i$, a power of 2]
		\item[\textbf{Initially:} $a_2, a_3$ co-located at $a_1 + (h, 0, \dots, 0)$, and $a_4$ at $a_1 + (X, 0, \dots, 0)$]
		\item[\textbf{On Termination:} $a_2, a_3$ co-located at $a_1 + (h, 0, \dots, 0)$, and $a_4$ at $a_1 + (hX, 0, \dots, 0)$]
		\While{\isdiv{2, a_1, a_2, a_3}} \Comment{Multiplication by $h$}
		\State \divstack{2, a_1, a_2, a_3}
		\State \multstack{3, a_1, a_2, a_3}
		\State Agent $a_3$ moves north until it meets $a_4$
		\State \multstack{2, a_1, a_3, a_4}
		\State Agent $a_3$ moves south until it meets $a_2$
		\EndWhile
		\While{\isdiv{3, a_1, a_2, a_3}} \Comment{Restore the value of $h$}
		\State \divstack{3, a_1, a_2, a_3}
		\State \multstack{2, a_1, a_2, a_3}
		\EndWhile
	\end{algorithmic}
\end{algorithm}

We now argue why \Cref{alg:multh} performs the desired operation:
Let $j = 1, \dots, i$ be the iteration of the first while loop.
At the end of the $j^{th}$ iteration, agents $a_2$ and $a_3$ will be co-located at distance $3^j 2^{i-j}$ from agent $a_1$, while agent $a_4$ will be at distance $2^j X$ from $a_1$.
Thus, after $i$ iterations, and remembering that $h = 2^i$, we have that $a_2$ and $a_3$ will be at distance $3^i$, and $a_4$ will be at distance $2^i X = hX$ from $a_1$.
Finally, with the same kind of argument, we also get that after the execution of the second while loop, agents $a_2$ and $a_3$ will have returned to distance $2^i = h$ from $a_1$.

As mentioned previously, the other stack operations can be performed analogously:
For \divstack{h}, we can just replace the calls to \multstack{2}, whenever the agents $a_1, a_3$, and $a_4$ act as the stack, with \divstack{2}.
For the check of divisibility by $h$, we can first execute the first while loop of the \divstack{h} routine.
If this leaves agent $a_2$ at distance $1$ from $a_1$, we know that $X$ is divisible by $h$, and otherwise we know that it is not.
To restore the positions of $a_2, a_3$ and $a_4$, we can use the second while loop of the \divstack{h} routine, where we additionally also add a call to \multstack{2, a_1, a_3, a_4} (with appropriate movement of the agent $a_3$ in between).

\paragraph{Runtime.}
Now that we have seen how to perform the necessary stack operations, we still need to argue that they can be performed efficiently.
Let $X_{\max}$ be the maximum distance between agent $a_1$ and $a_4$ during the execution of a stack operation.
Further, let us assume that throughout one operation, this distance is always much larger than the value of $h$.
As before we focus on the operation \multstack{h}:
In that case, $X_{\max}$ is the stack size at the end of the operation.
By the assumption that the stack size is always larger than $h$, the complexity of each iteration of the first while loop is $O(X')$, where $X'$ is the stack size at the end of the iteration.
As the stack size is doubled in each iteration, the complexity is dominated by its last term, which is $O(X_{\max})$.
Again, as we assume that $X_{\max} \gg h$, the second while loop does not change this asymptotic complexity.
Thus, this is the complexity of one stack operation in the synchronous model.

For the semi-synchronous model, we have that the agent responsible for synchronization moves $O(X)$ steps to perform one normal step on a stack of size $X$.
Thus, the complexity in the semi-synchronous model is $O(X_{\max}^2)$.

\paragraph{Exploring the Hypercube.}
Given that we now know how to perform stack operations for non-constant values, we can use the algorithm presented in Dobrev et al. \cite[Algorithm 7]{Dobrev2019} for exploring a hypercube of side length $h$.
We briefly sketch their ideas:
The goal is to explore all $n$-tuples from a $h$-ary alphabet, i.e., all cells $c = (c_1, \dots, c_n)$ such that $0 \leq c_i \leq h$ for all $i$.
The authors of~\cite{Dobrev2019} show that this can be achieved by visiting them in lexicographical order.
A tuple $(c_1, \dots, c_n)$ is encoded as a distance $((((c_1 \cdot h + c_2) \cdot h + c_3) \cdot h + c_4) \dots ) \cdot h + c_n$, and they show that the generalized stack operations are sufficient to enumerate and visit all tuples in a lexicographic order.

\paragraph{Finding the Treasure.}
The final algorithm by Dobrev et al.~\cite{Dobrev2019} starts with exploring a hypercube of side length $h = 2$ centered at the origin.
As long as the treasure is not found, $h$ is doubled and a hypercube of the now doubled side length $h$ is searched again.
As shown in~\cite{Dobrev2019}, the cost of the exploration is dominated by the exploration of the last hypercube, which is shown to have side length at most $4D$.
In this exploration, the maximum value encoded by our stack is $X = \Theta(D^n)$, where we have $h = D$.
In particular, this means that we have $X \gg h$, which implies that we can upper bound the runtime of one stack operation by $O(X)$ in the synchronous, and $O(X^2)$ in the semi-synchronous model.
Finally, they note that the hypercube consists of $D^n$ cells, for each of which we perform a constant number of stack operations.
Thus, the total complexity is $O(X^2)$ in the synchronous and $O(X^3)$ in the semi-synchronous model, where $X = \Theta(D^n)$.

Finally, the authors of \cite{Dobrev2019} compare this to $V(D)$ which is the volume of the $\ell_1$-ball containing all cells of distance at most $D$ from the origin.
As $V(D) = \Theta(D^n)$ (for constant dimension $n$), we finally get the following:
\thmpoly*

\section{Conclusion and Open Problems}

In this work, we studied the problem of exploring the $n$-dimensional oriented grid using autonomous agents controlled by deterministic finite automata.
We provided tight bounds for a number of settings:
For synchronous and semi-synchronous exploration, we showed that $3$, resp.\ $4$, agents are sufficient.
Further, we showed that under a natural assumption, we can get the same results even on unoriented grids.
We also made progress on the problem of finding a treasure in a number of steps that is polynomial in its distance from the starting point.
We showed that just one extra agent is enough for polynomial time exploration.

However, still a number of open questions remain:
Can we prove a higher lower bound on the number of agents required to explore an unoriented grid, than we can for oriented grids?
Is there a protocol that achieves both an optimal number of agents and polynomial time exploration?
For $n \geq 3$, can we improve the semi-synchronous protocol using randomness (the best known lower bound states that at least 3 agents are required)?
How much can we reduce the computational power of the agents without compromising the optimal bounds?
In our protocols, we can, e.g., replace agent $a_1$ with a movable marker, and in the semi-synchronous protocols, we can replace all agents except one with a movable marker; can we allow further/other restrictions?


\bibliographystyle{alpha}
\bibliography{space-ants}
%

\end{document}